\documentclass[a4paper,11pt,reqno]{amsart}
\usepackage[a4paper]{geometry}
\usepackage{amssymb,amsmath}
\usepackage{mathtools}
\usepackage{enumerate}
\usepackage{bookmark}
\usepackage{hyperref}
\usepackage[initials,backrefs]{amsrefs}

\numberwithin{equation}{section}
\theoremstyle{plain}
\newtheorem{theorem}{Theorem}
\newtheorem{lemma}{Lemma}
\theoremstyle{definition}
\newtheorem{definition}{Definition}
\newtheorem*{assi*}{(I) Short-range interaction}
\newtheorem*{asspone*}{(P1) Log-H\"older continuity condition}
\newtheorem*{assptwo*}{(P2) Rosenblatt strongly mixing condition}
\newtheorem*{dskn*}{$\dskn$}
\newtheorem*{dsknn*}{$\dsknn$}                                                                    
\theoremstyle{remark}

\newcommand{\prob}[1]{\DP\left\{#1\right\}}
\newcommand{\esm}[1]{\mathbb{E}\left[\,#1\,\right]}
\newcommand{\Bone}{\mathbf{1}}

\newcommand{\BC}{\mathbf{C}}

\newcommand{\BG}{\mathbf{G}}
\newcommand{\BH}{\mathbf{H}}

\newcommand{\BK}{\mathbf{K}}

\newcommand{\BP}{\mathbf{P}}
\newcommand{\BU}{\mathbf{U}}

\newcommand{\BX}{\mathbf{X}}

\newcommand{\CE}{\mathcal{E}}

\newcommand{\DN}{\mathbb{N}}
\newcommand{\DP}{\mathbb{P}}
\newcommand{\DR}{\mathbb{R}}
\newcommand{\DZ}{\mathbb{Z}}
\newcommand{\BDelta}{\mathbf{\Delta}}
\newcommand{\BPsi}{\mathbf{\Psi}}
\newcommand{\Bx}{\mathbf{x}}
\newcommand{\By}{\mathbf{y}}

\newcommand{\Bu}{\mathbf{u}}
\newcommand{\Bv}{\mathbf{v}}
\newcommand{\FB}{\mathfrak{B}}

\newcommand{\ess}{\mathrm{ess}}

\DeclareMathOperator{\card}{card}

\DeclareMathOperator{\dist}{dist}
\DeclareMathOperator{\supp}{supp}

\newcommand{\ee}{\mathrm{e}}

\newcommand{\condI}{\mathbf{(I)}}
\newcommand{\condPone}{\mathbf{(P1)}}
\newcommand{\condPtwo}{\mathbf{(P2)}}

\newcommand{\dskn}{\mathbf{(DS.}k,N\mathbf{)}}
\newcommand{\dsknn}{\mathbf{(DS.}k,n,N\mathbf{)}}

\newcommand{\tto}[1]{\smash{\mathop{\,\,\,\, \longrightarrow \,\,\,\, }\limits_{#1}}}

\begin{document}
\title[N-body localization at low energy for correlated potentials]{ N-body  localization  for the Anderson model with  strongly mixing correlated random potentials}

\author[T.~Ekanga]{Tr\'esor EKANGA$^{\ast}$}

\address{$^{\ast}$%
Institut de Math\'ematiques de Jussieu, 
Universit\'e Paris Diderot,
Batiment Sophie Germain,
13 rue Albert Einstein,
75013 Paris,
France}
\email{tresor.ekanga@imj-prg.fr}
\date{\today}

\begin{abstract}
This work establishes the Anderson localization in both the spectral exponential and the strong dynamical localization for the multi-particle Anderson tight-binding model with correlated but strongly mixing random external potential. The results are obtained near the lower edge of the spectrum of the multi-particle Hamiltonian. In particular, the exponential decay  of the eigenfunctions is proved in the max-norm and the dynamical localization in the Hilbert-Schmidt norm. The proofs need the conditional probability distribution  function of the  random external stochastic processes to obey the uniform log-H\"older continuity condition.
\end{abstract}

\maketitle

\section{Introduction, assumptions and the main result}

\subsection{Introduction}
We study the discrete multi-particle Anderson type Hamiltonians with correlated but strongly mixing random external potentials. In many earlier works on the multi-particle theory of many body interacting quantum systems \cites{AW09,CS09,KN13,E13}, the results were obtained for independent and identically distributed random external potentials. The method used by the authors was either an extension to multi-particle systems of the fractional moment method developed by Aizenmann and Molchanov \cite{AM93} in the single particle theory or an extension to multi-particle systems of multi-scale analysis developed by Fr\"ohlich and Spencer \cite{FMSS85}.

The analysis of correlated potentials was done by von Dreifus and Klein \cite{DK91} where the authors proved localization and for completely analytic Gibbs measures. Later, Chulaevsky \cites{BCS11,CS08,C08,C10} obtained the Wegner type bounds for multi-particles  models with correlated potentials using an extension to multi-particle systems of the so-called Stollmann's Lemma implemented in the book By Stollmann \cite{St01} for single-particle models in the continuum. Also, Klopp himself \cite{Kl12} analyzed the spectral statistics of multi-particle random operators with weakly correlated potentials. For the localization aspects of multi-particle models with correlated random potentials, the work by Chulaevsky \cite{C16} provides a proof of localization under the high disorder regime. Moreover, in \cite{C16}, the author assumed the Rosenblatt's strongly mixing condition for the correlated potentials.

In this paper,  we use the Wegner type bound obtained in \cite{C10} and prove localization at low energy for multi-particle models with correlated but strongly mixing random external potentials. For this result at low energy, we develop a fairly modified version of the multi-particle multi-scale analysis following the scheme proposed in our earlier work   \cite{E13} which was firstly introduced in the paper \cite{CS09}. There is an important problem that overcome for multi-particle systems, in particular the analysis near the bottom of the spectrum because the ergodic properties and  results cannot longer been applied as for single particle models. So, we have to develop new strategies in order to ensure the almost surely non-randomness of the multi-particle lower spectral edge because this is important for the multi-particle Anderson localization at low energy as it is mentioned in \cite{E13}.

The proofs of our localization results  are performed with the help of the multi-particle multi-scale analysis which in turn, uses both the uniform H\"older continuity of the conditional probability distribution function of the correlated random variables and the Rosenblatt's strongly mixing condition of the random field. The main novelty of the paper is that localization is obtained for energies near the bottom of the spectrum and for correlated potentials. The paper is in some way an extension of the i.i.d. case \cite{E13} and for that reason all the proofs using independence of the random variables are revisited and re-written.  

We now discuss on the structure of the rest of the paper. In the next Section, we present the general facts useful for the multi-scale analysis and proves the initial scale bounds. Section 3 is devoted to the scale induction step of the multi-scale analysis. Finally in Section 4 we sketch the proofs of the main results. Some parts of the rest of the paper overlap with \cite{E13}.

\subsection{The model and the assumptions}
We are interested in the $N$-particle Hamiltonian under some assumptions on $\BDelta$, $V$ and $\BU$. Fix an arbitrary integer $N\geq 2$. Since multi-scale analysis applied to this Hamiltonian requires also the consideration of Hamiltonians for $n$-particle subsystems for any $1\leq n\leq N$, we introduce notations for such $n$.
Given an integer $\nu\ge 1$, we generally equip $\DZ^{\nu}$ with the max-norm $|\,\cdot\,|$ defined by
\begin{equation}\label{eq:max.norm}
|x|=\max\{|x_1|,\dots,|x_{\nu}|\},
\end{equation}
for $x=\bigl(x_1,\dots,x_{\nu}\bigr)\in\DZ^{\nu}$. Occasionally (e.g., in the definition of the Laplacian)
we will use another norm defined by
\begin{equation}\label{eq:sum.norm}
|x|_1=|x_1|+\dots+|x_{\nu}|.
\end{equation}
Let $n\geq 1$ and $d\geq 1$ be two integers. A configuration of $n$ distinguishable quantum particles $\{x_1,\ldots,x_n\}$ in the lattice $\DZ^d$ is represented by a lattice vector $\Bx\in(\DZ^d)^n\cong \DZ^{nd}$ with coordinates $x_j=(x_j^{(1)},\dots,x_j^{(d)})\in\DZ^d$, $j=1,\dots,n$. 
 The $nd$-dimensional lattice nearest-neighbor Laplacian $\BDelta$ is defined by 
\begin{equation} \label{eq:def.Delta}
(\BDelta\BPsi)(\Bx)
=\sum_{\substack{\By\in\DZ^{nd}\\|\By-\Bx|_1=1}}\left(\BPsi(\By)-\BPsi(\Bx)\right)
=\sum_{\substack{\By\in\DZ^{nd}\\|\By-\Bx|_1=1}}\BPsi(\By)-2dn\BPsi(\Bx),
\end{equation}
for any $\BPsi\in\ell^2(\DZ^{nd})$, $\Bx\in\DZ^{nd}$. Note that $\BDelta$ is bounded and $-\BDelta$ is nonnegative.
 We will consider  random Hamiltonians $\BH^{(n)}(\omega)$ for $n=1,\ldots,N$ of the form
\[
\BH^{(n)}(\omega)=-\BDelta+\sum_{j=1}^nV(x_j,\omega)+\BU,
\]
acting in the Hilbert space $\mathcal{H}^{(n)}=\ell^2(\DZ^{dn})$.

\begin{assi*}
The potential of inter-particle interaction
\[
\BU\colon(\DZ^d)^n\to\DR
\]
is bounded and of the form
\begin{equation}\label{eq:def.U}
\BU(\Bx)=\sum_{1\leq i<j\leq n}\Phi(|x_i - x_j|),
\end{equation}
where the points $\{x_i,i=1,\ldots,n\}$ represent the coordinates of $\Bx\in(\DZ^d)^n$ and 
$\Phi\colon\DN \to \DR_+$
is a compactly supported non-negative function:
\begin{align}\label{eq:cond.U}
\exists\, r_0\in\DN:\quad
\supp\, \Phi \subset [0,r_0].
\end{align}
We will call $r_0$ the ``range'' of the interaction $\BU$.
\end{assi*}

The random field $\{V(x,\omega); x\in\DZ^d\}$ is measurable with respect to some probability space $(\Omega,\FB,\DP)$. We define 
\[
F_{V,x}(t):=\prob{V(x,\omega)\leq t}\qquad \text{ and } F_{V,x}(t\big|\FB_{\neq x}):=\prob{V(x,\omega)\leq t\big| \FB_{\neq x}},
\]
the conditional probability distribution functions of $V$ where $\FB_{\neq x}$ represents the sigma-algebra generated by the random variables $\{V(y,\omega); y\neq x\}$

\begin{asspone*}
It is assumed that the conditional distribution functions $F_{V,x}$ are uniformly Log-H\"older continuous: for some $\kappa>0$ and any $\varepsilon >0$,
\[
\ess \sup_{x\in\DZ^d} \sup_{t\in\DR} \left(F_{V,x}(t+\varepsilon\big|\FB_{\neq x})-F_{V,x}(t\big|\FB_{\neq x})\right)\leq Const\cdot |\ln(\varepsilon)|^{\kappa}.
\]
\end{asspone*}

\begin{assptwo*}
Let $L>0$ and positive constants $C_1>0$, $C_2>0$. For any pair of subsets $\Lambda', \Lambda'' \subset \DZ^{d}$ with $\dist(\Lambda',\Lambda'')\geq L$ and any events $\CE'\in\FB_{\Lambda'}$, $\CE''\in\FB_{\Lambda''}$,
\[
\left|\prob{\CE'\cap\CE''}-\prob{\CE'}\prob{\CE''}\right|\leq \ee^{-C_1\ln^2 L}.
\]
Further, for any integer $\ell\geq 2$, and random variables $X_1(\omega),\ldots X_{\ell}(\omega)$, we have that
\[
\left|\esm{X_1\cdots X_{\ell}} - \esm{X_1}\cdots \esm{X_{\ell}}\right| \leq \ee^{-C_2 L^d},
\]
with $C_2> 3^d$.
\end{assptwo*}
Assumption $\condPtwo$ was used by Chulaevsky in \cite{C10} in the framework of his so-called Direct scaling of the multi-scale analysis under the high disorder regime. Above, 
$\FB_{\Lambda'}$ and $\FB_{\Lambda''}$ are the sigma-algebra generated  by the random variables $\{V(x,\omega); x\in\Lambda'\}$ and $\{V(x,\omega); x\in\Lambda''\}$ respectively. 

\subsection{The result}

\begin{theorem}
Assume that the hypotheses $\condI$, $\condPone$ and $\condPtwo$ hold true. Then

\begin{enumerate}
\item[A)] 
The lower spectral edge $E_0^{(N)}$ of $\BH^{(N)}(\omega)$, is almost surely non-random and there exist $E^*> E_0^{(N)}$ such that the spectrum of $\BH^{(N)}(\omega)$ in $[E_0^{(N)}, E^*]$ is pure point and each eigenfunction corresponding to eigenvalues in $[E_0^{(N)},E^*]$ is exponentially decaying at infinity in the max-norm.\\

\item[B)]
There exist $E^*>E^{(N)}_0$ and $s^*(N,d)>0$ such that for any bounded domain $\BK\subset \DZ^{Nd}$ and any $0<s<s^*$ we have 
\begin{equation}\label{eq:low.energy.dynamical.loc}
\esm{\sup_{\|f\|_{\infty}\leq 1}\Bigl\| |\BX|^{\frac{s}{2}}f(\BH^{(N)}(\omega))\BP_{I}(\BH^{(N)}(\omega))\Bone_{\BK}\Bigr\|_{HS}^2}<\infty,
\end{equation}
where $(|\BX|\BPsi)(\Bx):=|\Bx|\BPsi(\Bx)$, $\BP_{I}(\BH^{(N)}(\omega))$ is the spectral projection of $\BH^{(N)}(\omega)$ onto the interval $I:=[E^{(N)}_0,E^*]$, and the supremum is taken over bounded measurable functions $f$. 
\end{enumerate} 
\end{theorem}

\section{Large deviation bounds}
We define the cube 
\[
C^{(1)}_L(0):=\{x\in \DZ^d: |x|\leq L\},
\]
and more general 
\[
\BC^{(n)}_L(\Bu)=\{\By\in\DZ^{nd}: |\By-\Bx|\leq L\},
\]
$n=1,\ldots,N$.
Introduce the property of \emph{non singularity}:
\begin{definition}
The cube $C^{(1)}_L(0)$ is called $(E,m)$-non singular iff
\[
\max_{y\in \partial^-C^{(1)}_L(0)}\left| \left(H_{C^{(1)}_L(0)}(\omega)-E\right)^{-1}\right|\leq \ee^{-\gamma(m,L,n)L},
\]
where $\gamma(m,L,n)=m(1+L^{-1/8})^{N-n+1}$ and $\partial^-C^{(1)}_L(0)=\{y\in\DZ^d: |x-y|=L\}$.
 Also define the Green functions 
$\BG^{(n)}_{\BC^{(n)}_L(0)}(E,\Bx,\By)$ the matrix elements of the restricted resolvent operator  with simple boundary conditions $(\BH^{(n)}_{\BC^{(n)}_L(0)}(\omega)-E)^{-1}$
\end{definition}
We prove in this subsection an analog of the large deviation estimate of \cite{K08} in the case of correlated potentials under the assumption $\condPtwo$. 

\begin{lemma}\label{lem:DEV} Under assumption $\condPtwo$, 
for $L=\lfloor (\beta E)^{-\frac{1}{2}}\rfloor$ with $\beta$ small and $L$ large enough,

\[
\prob{\frac{1}{|C^{(1)}_L(0)|} \sum_{x\in C^{(1)}_L(0)} V^{L}(x,\omega)<\frac{E}{2}}\leq \ee^{-\gamma|C^{(1)}_L(0)|},
\]
for some $\gamma>0$ and $V^{(L)}(x,\omega)=\min\{V(x,\omega), \frac{c}{3} L^{-2}\}$ with $c>0$.
\end{lemma}

\begin{proof}
We estimate
\begin{align}
&\prob{\frac{1}{|C^{(1)}_L(0)|} \sum V^{L}(i,\omega)<\frac{E}{2}}\notag\\
\leq \quad &\prob{\frac{1}{|C^{(1)}_L(0)|} \sum V^{L}(i,\omega)<\frac{\beta^2}{2}L^{-2}}\notag\\
\leq\quad &\prob{\card \{i\mid   V^{L}(i,\omega)<\frac{c}{3}L^{-2}\}\geq (1-\frac{3 \beta^2}{c})|C^{(1)}_L(0)|L}\label{eq:LD}.
\end{align}
Indeed, if less than $(1-\frac{3\beta^2}{c})|C^{(1)}_L(0)|$ of the $V(i)$ are below $\frac{c}{3}L^{-2}$ then more than $\frac{3\beta^2}{c}|C^{(1)}_L(0)|$ of them are at least $\frac{C}{3}L^{-2}$. In this case
\begin{align*}
\frac{1}{|C^{(1)}_L(0)|} \sum V^L(i,\omega)&\geq \frac{1}{|C^{(1)}_L(0)|}\frac{3\beta^2}{c}|C^{(1)}_L(0)|\frac{c}{3} L^{-2}\\
&= \frac{\beta^2}{2} L^{-2}.\\
\end{align*}

There is a $\gamma>0$ such that $q_i:=\prob{V(i,\omega)<\gamma}<1$. We set 
\[
\xi_i=\left\{ \begin{array}{rl}
1\quad & \text{ if } V(i,\omega)<\gamma,\\
0\quad & \text{otherwise}.
\end{array}
\right.
\]
One has that $\esm{\xi}=q_i$, put $r=1-\frac{3\beta^2}{c}$. By taking $\beta$ small, we can ensure that $q_i<r<1$. Then for sufficiently large $L>0$,
\begin{align}
\eqref{eq:LD}&\leq \prob{\card\{i| V^L(i,\omega)< \frac{c}{3} L^{-2}\}\geq r |C^{(1)}_L(0)|}\notag\\
& \prob{\card \{i| V^L(i,\omega)<\gamma\}\geq r |C^{(1)}_L(0)|}\notag\\
&\leq \prob{\frac{1}{|C^{(1)}_L(0)|}\sum \xi\geq r} \label{eq:LD2}.\\
\end{align}
Recall that 
\[
\prob{X>a}\leq \ee^{-ta}\esm{\ee^{tX}}, \qquad \text{for $t\geq 0$}.
\]
 Indeed,
\begin{align*}
\prob{X>a}&=\int \chi_{X>a}(\omega)d\DP(\omega)\\
&\leq \in\ee^{-ta}\ee^{tX}\chi_{X>a}(\omega)d\DP(\omega)\\
&\leq \int \ee^{-ta}\ee^{tX}d\DP.
\end{align*}
Therefore, we obtain using the Rosenblatt mixing condition $\condPtwo$ that 
\begin{align}
\eqref{eq:LD2}&\leq \ee^{-|C^{(1)}_L(0)|tr}\cdot\esm{\prod_{i\in C^{(1)}_L(0)}\ee^{t\xi_i}}\notag\\
&\leq \ee^{-|C^{(1)}_L(0)|rt}\left((\prod_{i\in C^{(1)}_L(0)} \esm{\ee^{t\xi_i}}) + \ee^{-C_2 L^{d}}\right) \notag\\
&\leq \ee^{-|C^{(1)}_L(0)|rt}\ee^{\ln \prod_{i\in C^{(1)}_L(0)}\esm{\ee^{t\xi_i}}}+\ee^{-|C^{(1)}_L(0)|rt-|C^{(1)}_L(0)|}\label{eq:LD3}.
\end{align}
We first compute the first term in the last equation, we have:  
\begin{align*}
\ee^{-|C^{(1)}_L(0)|rt +\ln\prod_{i\in C^{(1)}_L(0)}\esm{\ee^{t\xi_i}}}&\leq \ee^{-rt|C^{(1)}_L(0)| +\sum_{i\in C^{(1)}_L(0)} \ln(\esm{\ee^{t\xi_i}})}\\
&\leq \ee^{-|C^{(1)}_L(0)|\left( rt - \frac{1}{|C^{(1)}_L(0)|}\sum_{i\in C^{(1)}_L(0)} \ln(\esm{\ee^{t\xi_i}})\right)}\\
\end{align*}
Set $f(t)=rt - \frac{1}{|C^{(1)}_L(0)|} \sum_{i\in C^{(1)}_L(0)|} \ln(\esm{\ee^{t\xi_i}})$, we have 

\[
f'(t)=r-\frac{1}{|C^{(1)}_L(0)|}\sum_{i\in C^{(1)}_L(0)|} \frac{\esm{\xi_i\ee^{t\xi_i}}}{\esm{\ee^{t\xi_i}}},
\]
So,
\begin{align*}
f'(0)&=r-\frac{1}{| C^{(1)}_L(0)|}\sum_{i\in C^{(1)}_L(0)|} q_i\\
&= \frac{1}{| C^{(1)}_L(0)|} \sum_{i\in C^{(1)}_L(0)|} (r-q_i)>0, 
\end{align*}
so that $f'(0)>0$ and $f(0)=0$. Therefore, there exists $t_0>0$ such that $f(t_0)>0$ yielding some constant $\gamma_1= f(t_0)>0$.

We also bound the second term in the equation \eqref{eq:LD3} and obtain
\[
\ee^{-rt| C^{(1)}_L(0)|}\ee^{-C_2 L^d}\leq \ee^{-| C^{(1)}_L(0)|(rt+1)}=\ee^{-\gamma_2 | C^{(1)}_L(0)|},
\]
with $\gamma_2=rt+1>0$. Finally, setting $\gamma =\min\{\gamma_1,\gamma_2\}$ we obtain the desired bound and this competes the proof.
\end{proof} 

\begin{lemma} \label{lem:low.energy.gap.prob}
Let $H^{(1)}(\omega)=-\Delta+V(x,\omega)$ be a random single-particle lattice Schr\"o\-dinger operator in $\ell^2(\DZ^d)$. Assume that the random variables $V(x,\omega)$ satisfies condition $\condPtwo$. Then, for any $C>0$ there exist  arbitrary large $L_0(C)>0$ and $C_1,c>0$ such that for any cube $C^{(1)}_{L_0}(u)$, the lowest eigenvalue $E_0^{(1)}(\omega)$ of  $H_{C^{(1)}_{L_0}(u)}^{(1)}(\omega)$  satisfies
\begin{equation}\label{eq:initial.scale}
\DP\bigl\{E_0^{(1)}(\omega)\leq 2 CL_0^{-1/2}\bigr\}\leq C_1L_0^d \ee^{-cL_0^{d/4}}.
\end{equation}
\end{lemma}

\begin{proof}
The proof Combines equation (11.2) from the proof of Theorem 11.4 in \cite{K08} and the above result Lemma \ref{lem:DEV} on large deviations which in turn enables the Lifshitz asymptotics for single-particle models with correlated potentials.
\end{proof}
We now derive  the same result for the $n$-particle lattice Anderson model from  the results given by Theorem \ref{thm:CT} and Lemma \ref{lem:low.energy.gap.prob} in the same way as it is done in the works \cites{E11,E13}. For that reason the details of the proof  are omitted.

\begin{lemma} \label{lem:initial.scale.V.positive}
Let $\BH^{(n)}(\omega)=-\BDelta+V(x_1,\omega)+\dots+V(x_n,\omega)+\BU(\Bx)$ be an $n$-particle random Schr\"odinger operator in $\ell^2(\DZ^{nd})$ where  $\BU$ and $V$ satisfy $\condI$, $\condPone$ and $\condPtwo$ respectively.
Then, for any $C>0$ there exist arbitrary large $L_0(C)>0$ and $C_1,c>0$ such that for any cube $\BC^{(n)}_{L_0}(\Bu)$ the lowest eigenvalue $E_0^{(n)}(\omega)$ of   $\BH_{\BC_{L_0}^{(n)}(\Bu)}^{(n)}(\omega)$ satisfies
\begin{equation}\label{eq:initial.bound}
\DP\bigl\{E_0^{(n)}(\omega)\leq 2CL_0^{-1/2}\bigr\}\leq C_1L_0^d\ee^{-cL_0^{1/4}}.
\end{equation}
\end{lemma}

\begin{proof}
The proof of Lemma \ref{lem:initial.scale.V.positive} becomes straightforward given the assertion of Lemma \ref{lem:low.energy.gap.prob}.
\end{proof}

\section{The multi-particle multi-scale analysis} 

It is convenient here to recall the Combes-Thomas estimate.

\begin{theorem}[Combes-Thomas estimate]\label{thm:CT}
Consider a lattice  Schr\"odinger operator
\[
H_\Lambda = -\Delta_\Lambda + W(x)
\]
acting in $\ell^2(\Lambda)$, $\Lambda\subset\DZ^{\nu}$, $\nu\geq 1$, with an arbitrary \footnote{This includes the cases of single- and multi-particle operators, that differ only by their potentials.} potential $W\colon\Lambda\to\DR$.
Suppose that $E\in\DR$ is such that\footnote{Theorem 11.2 from \cite{K08} is formulated with the equality  $\dist(E,\sigma(H_{\Lambda})) = \eta$, but it is clear from the proof that it remains valid if $\dist(E,\sigma(H_{\Lambda})) \ge \eta$.} $\dist(E,\sigma(H_{\Lambda})) \ge \eta$ with $\eta \in(0,1]$. 
Then 
\begin{equation}\label{eq:CT}
\forall\; x,y\in\Lambda\qquad
\left|\left(H_{\Lambda} -E\right)^{-1}(x,y)\right| \le 2 \eta^{-1}\,\ee^{-\frac{\eta}{12\nu}|x-y|}.
\end{equation}
\end{theorem}

\begin{proof}
See the proof of Theorem 11.2 \cite{K08}.
\end{proof}

\begin{theorem}\label{thm:initial.estimate}
Under assumptions $\condI$, $\condPone$ and $\condPtwo$ for any $p>0$, there exists arbitrarily large $L_0(N,d,p)$ such that if  $m:=(12Nd)L_0^{-1/2}$ and $E^*:=(12Nd)( 2^{N+1}m)$, then 
\[
\prob{\text{$\exists E\in (-\infty,E^*]; C^{(n)}_{L_0}(0)$ is $(E,m)$-S}}\leq L_0^{-2p4^{N-n }},
\]
for any $n=1,\ldots,N$.
\end{theorem}

\begin{proof}
Set $C=(12Nd)^2\cdot2^{N+1}$ and let $\BC^{(n)}_{L_0}(0)$ be a cube in $\DZ^{nd}$. Consider  $\omega\in\Omega$ such that the first eigenvalue $E_0^{(n)}(\omega)$ of $\BH^{(n)}_{\BC^{(n)}_{L_0}(0)}(\omega)$ satisfies $E_0^{(n)}(\omega)>2CL_0^{-1/2}$. Then for all $E\leq CL_0^{-1/2}=E^*$ we have
\[
\dist(E,\sigma(\BH^{(n)}_{\BC^{(n)}_{L_0}(0)}(\omega)))=E_0^{(n)}(\omega) -E >CL_0^{-1/2}=: \eta.
\]
\noindent
For $L_0$ large enough $\eta\leq 1$ and the Combes-Thomas estimate (Theorem \ref{thm:CT}) implies that
 for any $\Bv\in\partial^-\BC^{(n)}_{L_0}(0)$:
\[
|\BG_{\BC^{(n)}_{L_0}(0)}(E,0,\Bv)|\leq 2 \eta^{-1}\, \ee^{-\frac{CL_0^{-1/2}}{12Nd}L_0}\\.
\]
Now observe that $\frac{CL_0^{-1/2}}{12 Nd}=2^{N+1}m$. Thus
\begin{align*}
|\BG_{\BC^{(n)}_{L_0}(0)}(E,0,\Bv)|&\leq 2 C^{-1}L_0^{1/2}\,  \ee^{-2^{N+1} m L_0}\\
&\leq  2 C^{-1}L_0^{1/2} \ee^{-2\gamma(m,L_0,n)L_0}
\\
&\leq \ee^{-\gamma(m,L_0,n)L_0},
\end{align*}
for $L_0$ large enough, since
\[ 
\gamma(m,L,n)=m(1+L^{-1/8})^{N-n+1}< m\times 2^N.
\]
This implies that $\BC^{(n)}_{L_0}(0)$ is $(E,m)$-NS. Since this holds assuming $E^{(n)}_0(\omega)>2CL_0^{-1/2}$, we get
\[
\prob{\text{$\exists E\leq E^*$, $\BC^{(n)}_{L_0}(0)$ is $(E,m)$-S}}\leq \prob{E_0^{(n)}(\omega)\leq 2CL_0^{-1/2}},
\]
and by Lemma \ref{lem:initial.scale.V.positive},
\[
\DP\bigl\{E_0^{(n)}(\omega)\leq 2CL_0^{-1/2}\bigr\}\leq C_1L_0^d\ee^{-cL_0^{1/4}}.
\]
Finally, the quantity $C_1 L_0^d\ee^{-cL_0^d/4}$ for $L_0$ sufficiently large is less than $L_0^{-2p\,4^{N-n}}$. This proves the claim since the probability for two cubes to be singular at the same energy is bounded by the probability of either one of them to be singular.
\end{proof}

Now the rest of the multi-particle multi-scale analysis can be done exactly in the same way as in our earlier work \cite{E13} in the case of i.i.d. random potential.    

\section{Proof of the main result}

\subsection{Almost surely non-randomness of the lower spectral edges}

By the single-particle theory for ergodic random Schr\"odinger operators $H^{(1)}(\omega)$, see for example \cite{CL90}, there exists a closed subset $\Sigma\subset \DR$ such that $\sigma(H^{(1)}(\omega))=\Sigma$ with $\DP$-probability $1$. By hypothesis $\condPone$, the random variables $\{V(x,\omega)\}$ are almost surely non negative, so that $\Sigma \subset [0,+\infty)$. Moreover, by adding a constant to the random potential we can assume without loss of generality that, $\inf \Sigma =0$. Thus, $0\in\sigma(H^{(1)}(\omega))$ almost surely. We claim that it is also true for each multi-particle random Hamiltonian $\BH^{(n)}(\omega)$, $n=1,\ldots,N$.

Let $n=1,\ldots,N$. Assumption $\condI$ implies that $\BU$ is non negative and condition $\condPone$ implies that the random external potential $V$ is almost surely  non negative. So, since $-\BDelta\geq 0$, we have that $\sigma(\BH^{(n)}(\omega)\subset[0,+\infty)$ almost surely.

Let $k,m \in\DN$. Define,
\[
B_{k,m}:=\{\Bx\in\DZ^{nd}: \min_{i\neq j}|x_i-x_j|>r_0+2km\}
\]
where $r_0>0$, is the range of the interaction $\BU$. We also define the following sequence in $\DZ^{nd}$,
\[
\Bx^{k,m}:=C_{k,m}(1,\ldots,nd),
\]
where $C_{k,m}=r_0+2km+1$. Using the identification $\DZ^{nd}\cong (\DZ^d)^n$, we can also write $\Bx^{k,m}=C_{k,m}(x_1^{k,m},\ldots,x_n^{k,m})$ with each $x_i^{k,m}\in\DZ^d$, $i=1,\ldots,n$. Obviously, each term $\Bx^{k,m}$ of the sequence $(\Bx^{k,m})_{k,m}$ belongs to $B_{k,m}$. For $j=1,\ldots,n$, set,

\[
H^{(1)}_j(\omega):=-\Delta + V(x_j;\omega).
\]
We have that almost surely $\sigma(H^{(1)}_j(\omega))=[0;+\infty)$, see for example \cite{St01}. So, if we set for $j=1,\ldots, n$,
\[
\Omega_j:=\{\omega\in\Omega: \sigma(H^{(1)}_j(\omega))=[0,+\infty)\},
\]
$\prob{\Omega_j}=1$ for all $j=1,\ldots,n$. Now, put 
\[
\Omega_0:= \bigcap_{j=1}^n \Omega_j.
\]
We also have that $\prob{\Omega_0)}=1$. Let $\omega\in\Omega_0$, for this $\omega$,   by the Weyl criterion, there exist $n$ Weyl sequences $\{(\phi^m_j)_m: j=1,\ldots,n\}$ related to $0$ and each operator $H^{(1)}_j(\omega)$. By the density property of compactly supported functions $C^{\infty}_c(\DR^d)$ in $L^2(\DR^d)$, we can directly assume that each $\phi_j^m$ is of compact support, i.e., $\supp \phi_j^m\subset C^{(1)}_{k_j m}(0)$ for some integer $k_j$ large enough. Set 
\[
k_0:=\max_{j=1,\ldots,n} k_j,
\]
and put, $\Bx^{k_0,m}=(x_1^{k_0,m},\ldots, x_n^{k_0,m})\in B_{k_0,m}$. We translate each function $\phi_j^m$ to have support contained in the $C^{(1)}_{k_0m}(x_j^{k_0})$. Next, consider the sequence $(\phi^m)_m$ defined by the tensor product,
\[
\phi^m:=\phi^m_1\otimes\cdots\otimes \phi^m_n. 
\]
We have that $\supp \phi^m\subset \BC^{(n)}_{k_0m}(\Bx^{k_0,m})$ and we aim to show that, $(\phi^m)_m$ is a Weyl sequence for $\BH^{(n)}(\omega)$ and $0$. For any $\By\in\DZ^{nd}$:

\[
|(\BH^{(n)}(\omega)\phi^m)(\By)|= |(\BH^{(n)}_0(\omega)\phi^m)|.
\]
Indeed, for the values of $\By$ inside the cube $\BC^{(n)}_{k_0,m}(\Bx^{k_0,m})$ the interaction potential $\BU$ vanishes and for those values outside that cube, $\phi^m$ equals zero too. Therefore, 
\begin{align*}
\|\BH^{(n)}(\omega)\phi^m\|&\leq \| \BH^{(n)}_0(\omega)\phi^m\|\\
&\leq \sum_{j=1}^n \| (H^{(1)}_j(\omega))\phi_j^m\| \tto{m\to +\infty} 0,
\end{align*}
because, for all $j=1,\ldots,n$, $\|(H^{(1)}_j(\omega))\phi^m_j\|\rightarrow 0$ as $m\rightarrow +\infty$, since $\phi^m_j$ is a Weyl sequence for $H^{(1)}_j(\omega)$ and $0$. This completes the proof.

\subsection{Proof of the localization results}
See the proofs of Theorems 2 and 3 in \cite{E13}.

\end{document}